\documentclass[11pt,letterpaper]{article}

\usepackage{amsmath,amsfonts,amsthm,amssymb,stmaryrd,graphicx,relsize,bm}  
\theoremstyle{plain}

\numberwithin{equation}{section}

\newtheorem{thm}{Theorem}[section]
\newtheorem{lem}[thm]{Lemma}
\newtheorem{cor}[thm]{Corollary}

\newenvironment{exam}[1]
{\begin{flushleft}\textbf{Example #1}.\enspace}
{\end{flushleft}}

\allowdisplaybreaks  

\newcounter{cond}

\newcommand{\complex}{{\mathbb C}}
\newcommand{\real}{{\mathbb R}}

\newcommand{\tbullet}{\mathrel{\raise .4ex\hbox{\tiny$\bullet$}}} 

\newcommand{\rmtr}{\mathrm{tr\,}}
\newcommand{\rmin}{\mathrm{In\,}}

\newcommand{\rmob}{\mathrm{Ob\,}}

\newcommand{\cscript}{\mathcal{C}}
\newcommand{\escript}{\mathcal{E}}
\newcommand{\iscript}{\mathcal{I}}
\newcommand{\jscript}{\mathcal{J}}
\newcommand{\kscript}{\mathcal{K}}
\newcommand{\lscript}{\mathcal{L}}

\newcommand{\oscript}{\mathcal{O}}
\newcommand{\pscript}{\mathcal{P}}
\newcommand{\sscript}{\mathcal{S}}

\newcommand{\cscripthat}{\widehat{\cscript}}
\newcommand{\iscripthat}{\widehat{\iscript}}
\newcommand{\jscripthat}{\widehat{\jscript}}
\newcommand{\kscripthat}{\widehat{\kscript}}

\newcommand{\uhat}{\widehat{u}}
\newcommand{\zerohat}{\widehat{0}}
\newcommand{\iscriptbar}{\overline{\iscript}}
\newcommand{\jscriptbar}{\overline{\jscript}}

\newcommand{\ab}[1]{\left|#1\right|}
\newcommand{\brac}[1]{\left\{#1\right\}}
\newcommand{\paren}[1]{\left(#1\right)}
\newcommand{\sqbrac}[1]{\left[#1\right]}
\newcommand{\elbows}[1]{{\left\langle#1\right\rangle}}
\newcommand{\ket}[1]{{\left|#1\right>}}
\newcommand{\bra}[1]{{\left<#1\right|}}

\begin{document}

\title{SEQUENTIAL PRODUCTS OF\\QUANTUM MEASUREMENTS}
\author{Stan Gudder\\ Department of Mathematics\\
University of Denver\\ Denver, Colorado 80208\\
sgudder@du.edu}
\date{}
\maketitle

\begin{abstract}
Our basic structure is a finite-dimensional complex Hilbert space $H$. We point out that the set of effects on $H$ form a convex effect algebra. Although the set of operators on $H$ also form a convex effect algebra, they have a more detailed structure. W introduce sequential products of effect and operations. Although these have already been studied, we introduce the new concept of sequential products of effects with operations and operations with effects. We then consider various special types of operations. After developing properties of these concepts, the results are generalized to include observables and instruments. In particular, sequential products of observables with instruments and instruments with observables are developed. Finally, we consider conditioning and coexistence of observables and instruments.
\end{abstract}

\section{Basic Definitions}  
Let $H$ be a finite-dimensional complex Hilbert space.The basic concepts of the present version of quantum measurement theory are the sets of effects $\escript (H)$, states $\sscript (H)$ and operations $\oscript (H)$ on $H$ \cite{bgl95,hz12,kra83,nc00}. Denoting the set of linear operators on $H$ by
$\lscript (H)$, we define $\escript (H)=\brac{a\in\lscript (H)\colon 0\le a\le I}$ where $0,I$ are the zero and identity operators, respectively. An effect
$a\in\escript (H)$ corresponds to a two-valued $yes-no$ experiment. The effect $a'=I-a$ has the value $yes$ if $a$ has the value $no$ and we call $a'$ the \textit{complement} of $a$. If $a,b\in\escript (H)$ and $a+b\in\escript (H)$, then we write $a\perp b$. When $a\perp b$, we interpret
$a+b\in\escript (H)$ as the parallel sum of the effects $a$ and $b$. The effect $0$ always has value $no$ and $I$ always has value $yes$. The fact that $a\perp a'$ and $a+a'=I$ indicates that either $a$ or $a'$ has the value $yes$ but not both. An element $\rho\in\escript (H)$ that satisfies
$\rmtr (\rho )=1$ is called a \textit{state}. States describe the initial condition of a quantum system. If $a\in\escript (H)$ and $\rho\in\sscript (H)$, then the \textit{probability} that $a$ has value $yes$ when the system is in state $\rho$ is given by $\pscript _\rho (a)=\rmtr (\rho a)$. It is clear that
$\pscript _\rho (0)=0$, $\pscript _\rho (I)=1$ and if $a\perp b$, then $\pscript _\rho (a+b)=\pscript _\rho (a)+\pscript _\rho (b)$.

A linear map $\iscript\colon\lscript (H)\to\lscript (H)$ is \textit{completely positive}  if $\iscript\otimes I\colon H\otimes K\to H\otimes K$ is positive for every auxiliary finite dimensional Hilbert space $K$. We call $\iscript\colon\lscript (H)\to\lscript (H)$ an \textit{operation} if $\iscript$ is completely positive and $\rmtr\sqbrac{\iscript (\rho )}\le\rmtr (\rho )$ for every positive $\rho\in\lscript (H)$ \cite{bgl95,hz12,kra83,nc00}. An operation $\iscript$ is called a \textit{channel} if $\rmtr\sqbrac{\iscript (\rho )}=\rmtr (\rho )$ for every positive $\rho\in\lscript (H)$. In particular, if $\rho\in\sscript (H)$ then
$\rmtr\sqbrac{\iscript (\rho )}=1$. Any $\iscript\in\oscript (H)$ has a \textit{Kraus decomposition} $\iscript (\rho )=\sum\limits _{i=1}^nA_i\rho A_i^*$ where $A_i\in\lscript (H)$, $i=1,2,\ldots ,n$ \cite{bgl95,hz12,kra83,nc00}. The Kraus operators $A_i$ need not be unique. Since
$\rmtr\sqbrac{\iscript (\rho )}\le\rmtr (\rho )$ we have that
\begin{equation}                
\label{eq11}
\rmtr (\rho\sum A_i^*A_i)=\sum\rmtr (\rho A_i^*A_i)\!=\!\sum\rmtr (A_i\rho A_i^*)=\rmtr\sqbrac{\iscript (\rho )}\le\rmtr (\rho I)
\end{equation}
for every positive $\rho$. It follows that $\sum A_i^*A_i\le I$. An operation is thought of as an apparatus $\iscript$ that can be employed to measure an effect $\iscripthat$. We define the \textit{probability} that $\iscript$ has value $yes$ in the state $\rho$ to be
$\pscript _\rho (\iscript )=\rmtr\sqbrac{\iscript (\rho )}$ and $\iscripthat$ is the unique effect satisfying $\rmtr (\rho\iscripthat\,)=\pscript _\rho (\iscript )$ for all $\rho\in\sscript (H)$ and say that $\iscript$ \textit{measures} $\iscripthat$. Although every $\iscript\in\oscript (H)$ measures a unique
$\iscripthat\in\escript (H)$, an effect is measured by many operations. That is, many apparatuses can be employed to measure an effect $a$. Moreover,
$\iscript$ gives more information than $\iscripthat$ because $\iscript (\rho )/\rmtr\sqbrac{\iscript (\rho )}$ is the updated state after $\iscript$
(or $\iscripthat$) is measured (assuming $\rmtr\sqbrac{\iscript (\rho )}\ne 0$). One way to specify $\iscripthat$ is the following: If $\iscript$ has Kraus decomposition $\iscript (\rho )=\sum A_i\rho A_i^*$, then by \eqref{eq11} we have that $\iscripthat =\sum A_i^*A_i$. This also shows that if
$\iscript (\rho )=\sum B_j\rho B_j^*$ is another Kraus decomposition for $\iscript$, then $\sum B_j^*B_j=\sum A_i^*A_i$.

We now consider the mathematical structures of $\escript (H)$ and $\oscript (H)$. Let $(V,\le )$ be a finite-dimensional ordered real linear space
\cite{hz12,nc00} and let $u\in V$ satisfy $u>0$. Letting $E$ be the ordered interval
\begin{equation*} 
E=\sqbrac{0,u}=\brac{x\in V\colon 0\le x\le u}
\end{equation*}
we call $(E,0,u)$ a \textit{convex effect algebra} \cite{gud320}. For $x,y\in E$ we write $x\perp y$ if $x+y\in E$. It is easy to check that $E$ satisfies the effect algebra axioms:
\begin{list} {(\arabic{cond})}{\usecounter{cond}
\setlength{\rightmargin}{\leftmargin}}
\item If $x\perp y$, then $y\perp x$ and $x+y=y+x$.
\item If $y\perp z$ and $x\perp (y+z)$, then $x\perp y$, $z\perp (x+y)$ and we have $x+(y+z)=(x+y)+z$.
\item For every $x\in E$ there exists a unique $x'\in E$ such that $x\perp x'$ and $x+x'=u$.
\item If $x\perp u$, then $x=0$.
\end{list}
Moreover, $E$ is \textit{convex} because if $x_1,\ldots ,x_n\in E$ and $\lambda _i\ge 0$ with $\sum\lambda _i=1$, then $\sum\lambda _ix_i\in E$. Indeed, we have that
\begin{equation*} 
0\le \sum\lambda _ix_i\le\sum\lambda _iu=u
\end{equation*}
We call the element $x'$ in (3) the \textit{complement} of $x$. If $E=\sqbrac{0,u}$ and $F=\sqbrac{0,v}$ are convex effect algebras, a map $J\colon E\to F$ is an \textit{isomorphism} if $J$ is an order-preserving bijection, $J(u)=v$, $J(x+y)=J(x)+J(y)$ whenever $x\perp y$ and
$J(\sum\lambda _ix_i)=\sum\lambda _iJ(x_i)$ whenever $\lambda _i\ge 0$, $\sum\lambda _i=1$.

Notice that $\escript (H)$ is a convex effect algebra because $\escript (H)=\sqbrac{0,I}\subseteq\lscript _\real (H)$ where $\lscript _\real (H)$ is the real linear span of $\escript (H)$. In a similar way, letting $V$ be the real linear span of $\oscript (H)$ we have that $\oscript (H)=\sqbrac{0,I}\subseteq V$ where $0$ and $I$ are the zero and identity operations, respectively, so $\oscript (H)$ is also a convex effect algebra. However, this convex effect algebra does not specify the detailed probabilistic structure of $\oscript (H)$. Although $I\in\escript (H)$ is the unique certain effect, the identity
$I\in\oscript (H)$ is not the unique certain operation. If $\cscript$ is a channel, then for any $\rho\in\sscript (H)$ we have that
$\pscript _\rho (\cscript )=\rmtr\sqbrac{\cscript (\rho )}=1$ so $\cscript$ has value $yes$ for any state so $\cscript$ is certainly $yes$. Similarly, if
$\iscript ,\jscript\in\oscript (H)$ and $\iscripthat =\jscripthat$ then
\begin{equation*} 
\pscript _\rho (\iscript )=\rmtr\sqbrac{\iscript (\rho )}=\rmtr (\rho\iscripthat\,)=\rmtr (\rho\jscripthat\,)=\rmtr\sqbrac{\jscript (\rho )}=\pscript _\rho (\jscript )
\end{equation*}
for every $\rho\in\sscript (H)$ so we cannot distinguish $\iscript$ and $\jscript$ probabilistically. For $\iscript ,\jscript\in\oscript (H)$ we write
$\iscript\approx\jscript$ if $\iscripthat =\jscripthat$. It is clear that $\approx$ is an equivalence relation in $\oscript (H)$. We denote the equivalence class containing $\iscript$ by $\sqbrac{\iscript}$ and the set of equivalence classes by $\oscript (H)/\approx$. We write
$\sqbrac{\iscript}^\wedge =\iscripthat$ and it is clear that this is well-defined.

\begin{thm}    
\label{thm11}
$\oscript (H)/\negthickspace\approx\ $ is a convex effect algebra and ${\ }^\wedge\colon\oscript (H)/\negthickspace\approx\,\to\escript (H)$ is an isomorphism.
\end{thm}
\begin{proof}
We write $0\!=\!\brac{0}\!=\!\sqbrac{0}$ and $u=\sqbrac{\cscript}$ where $\cscript$ is a channel and we have that $\zerohat\!=\!0$ and $\uhat =I$. If
$\iscript ,\jscript\in\oscript (H)$ with $\iscript\perp\jscript$ we define $\sqbrac{\iscript}+\sqbrac{\jscript}=\sqbrac{\iscript +\jscript}$ and if 
$\iscript _i\in\oscript (H)$, $\lambda _i\ge 0$ with $\sum\lambda _i=1$, $i=1,2,\ldots ,n$ we define
\begin{equation*}
\sum _i\lambda _i\sqbrac{\iscript _i}=\sqbrac{\sum\lambda _i\iscript _i}
\end{equation*}
Moreover, we write $\sqbrac{\iscript}\le\sqbrac{\jscript}$ if $\iscripthat\le\jscripthat$. It is clear that these are well-defined. We then have that
$0\le\sqbrac{\iscript}\le u$ for all $\iscript\in\oscript (H)$. It is straightforward to show that $\oscript (H)/\negthickspace\approx\,=\sqbrac{0,u}$ is an order interval for a real ordered linear space $V$ consisting of the linear span of $\oscript (H)/\negthickspace\approx$ and hence is a convex effect algebra. Since
$\sqbrac{\iscript +\jscript}^\wedge =\iscripthat +\jscripthat$ and
\begin{equation*}
\paren{\sum\lambda _i\sqbrac{\iscript _i}}^\wedge =\paren{\sum\lambda _i\iscript _i}^\wedge =\sum\lambda _i\iscripthat _i
\end{equation*}
we have that ${\ }^\wedge$ preserves $+$ and convex combinations. Clearly, $\uhat =I$ and ${\ }^\wedge$ is order-preserving. If $a\in\escript (H)$ we will show there exists an $\iscript\in\oscript (H)$ such that $\iscripthat =a$ so ${\ }^\wedge$ is surjective. To show that ${\ }^\wedge$ in injective, suppose that $\sqbrac{\iscript}^\wedge =\sqbrac{\jscript}^\wedge$. Then $\iscripthat =\jscripthat$ so $\sqbrac{\iscript}=\sqbrac{\jscript}$. Hence,
${\ }^\wedge\colon\oscript (H)/\negthickspace\approx\,\to\escript (H)$ is an isomorphism.
\end{proof}

Besides sums and convex combinations, there is another important way of combining effects and operations. For $a,b\in\escript (H)$ we define their \textit{sequential product} $a\circ b=a^{1/2}ba^{1/2}$ where $a^{1/2}$ is the unique positive square-root of $a$ \cite{gg02,gn01}. It is easy to check that $a\circ b\le a$ so we indeed have that $a\circ b\in\escript (H)$. We interpret $a\circ b$ to be the effect resulting from first measuring $a$ and then measuring $b$. Then $\escript (H)$ has the mathematical structure of a convex, sequential effect algebra \cite{gud220,gud320}. Because of the sequential order of $a\circ b$, the measurement of $a$ can influence (interfere) with the measurement of $b$ but not vice versa. This is emphasized by the fact that $a\circ (b+c)=a\circ b+a\circ c$ and $a\circ\paren{\sum\lambda _ib_i}=\sum\lambda _ia\circ b_i$. We then say that $a\circ b$ is
\textit{additive} and \textit{affine} in the second variable. In general, $a\circ b$ is not additive or affine in the first variable. Also, it can be shown that
$a\circ b=b\circ a$ if and only if $ab=ba$ \cite{gn01}. In a similar way, for $\iscript ,\jscript\in\oscript (H)$ we define their
\textit{sequential product} $\iscript\circ\jscript (\rho )=\jscript\paren{\iscript (\rho )}$ \cite{gud120,gud320}. Although we still retain the same influence interpretation, $\iscript\circ\jscript$ is additive and affine in both variables. However, $\iscript\circ\jscript\ne\jscript\circ\iscript$, in general. We say that
$a\in\escript (H)$ is \textit{sharp} if $a$ is a projection. It can be shown that $a$ is sharp if and only if $a\wedge a'=0$. That is if $b\le a,a'$ then $b=0$. An effect $a$ is \textit{atomic} if $a$ is a one-dimensional projection. The next result pertains to additivity and commutativity of sharp and atomic effects.

\begin{thm}    
\label{thm12}
{\rm{(i)}}\enspace If $a_1,a_2,\ldots ,a_n$ are sharp and $\sum a_i=I$, then $b=\sum a_i\circ b$ if and only if $ba_i=a_ib$, $i=1,2,\ldots ,n$.
{\rm{(ii)}}\enspace If $a$ and $b$ are atomic and have the form $a=\ket{\phi}\bra{\phi}$, $b=\ket{\psi}\bra{\psi}$, then
$\pscript _\rho (a\circ b)=\pscript _\rho (b\circ a)$ if and only if $\elbows{\phi ,\rho\phi}=\elbows{\psi ,\rho\psi}$ or $ab=0$.
\end{thm}
\begin{proof}
(i)\enspace If $ba_i=a_ib$, $i=1,2,\ldots ,n$, then
\begin{equation*}
\sum a_i\circ b=\sum a_iba_i=\sum a_ib=b
\end{equation*}
Conversely, suppose that $b=\sum a_i\circ b=\sum a_iba_i$. Since $\sum a_i=I$, we obtain
\begin{equation*}
a_j+\sum _{i\ne j}a_ja_ia_j=a_j
\end{equation*}
Hence, $\sum\limits _{i\ne j}a_ja_ia_j=0$ so that $a_ja_ia_j=0$, $i\ne j$. Therefore,
\begin{equation*}
(a_ja_i)(a_ja_i)^*=a_ja_ia_ia_j=a_ja_ia_j=0
\end{equation*}
It follows that $a_ja_i=0$ for $i\ne j$. But then
\begin{equation*}
a_jb=a_jba_j=ba_j
\end{equation*}
for $j=1,2,\ldots ,n$.
(ii)\enspace We have that $\pscript _\rho (a\circ b)=\pscript _\rho (b\circ a)$ if and only if
\begin{equation*}
\rmtr\paren{\rho\ket{\phi}\bra{\phi}\,\ket{\psi}\bra{\psi}\,\ket{\phi}\bra{\phi}}=\rmtr\paren{\rho\ket{\psi}\bra{\psi}\,\ket{\phi}\bra{\phi}\,\ket{\psi}\bra{\psi}}
\end{equation*}
This is equivalent to
\begin{align*}
\ab{\elbows{\phi ,\psi}}^2\elbows{\phi ,\rho\phi}&=\ab{\elbows{\phi ,\psi}}^2\rmtr\paren{\rho\ket{\phi}\bra{\phi}}
  =\ab{\elbows{\phi ,\psi}}^2\rmtr\paren{\rho\ket{\psi}\bra{\psi}}\\
  &=\ab{\elbows{\phi ,\psi}}^2\elbows{\psi ,\rho\psi}
\end{align*}
Since $\elbows{\phi ,\psi}=0$ if and only if $ab=0$, then the result follows.
\end{proof}

It is easy to check that $a\in\escript (H)$ is atomic if and only if $b\le a$ implies that $b=\lambda a$ for some $\lambda\in\sqbrac{0,1}$.

\section{Bayes' Rules}  
First note that $\pscript _\rho$ has the usual properties of a probability measure on $\escript (H)$. That is, $0\le\pscript _\rho (a)\le 1$ for all
$a\in\escript (H)$, $\pscript _\rho (I)=1$ and if $a\perp b$, then $\pscript _\rho (a+b)=\pscript _\rho (a)+\pscript (b)$. This last equation shows that 
$a\perp b$ is the analogue of disjointness of events in classical probability theory. For $a,b\in\escript (H)$, it is natural to define the
\textit{conditional probability}
\begin{equation*} 
\pscript _\rho\paren{b\mid a}=\frac{\pscript _\rho (a\circ b)}{\pscript _\rho (a)}=\frac{\rmtr\paren{\rho a^{1/2}ba^{1/2}}}{\rmtr (\rho a )}
\end{equation*}
whenever $\pscript _\rho (a)\ne 0$. Although $\pscript _\rho\paren{\tbullet\mid a}$ satisfies the above three conditions for a probability measure, it does not satisfy Bayes' rules. Bayes' first rule says that if $a_1+a_2+\cdots +a_n=I$, then for all $b\in\escript (H)$ we have that
\begin{equation}                
\label{eq21}
\pscript _\rho (b)=\sum _{i=1}^n\pscript (a_i)\pscript _\rho (b\mid a_i)
\end{equation}
If \eqref{eq21} holds for all $\rho\in\sscript (H)$ it follows that $b=\sum a_i\circ b$. But according to Theorem~\ref{thm12}(i) this does not hold, in general.

Bayes' second rule says that
\begin{equation}                
\label{eq22}
\pscript _\rho (b)\pscript _\rho (a\mid b)=\pscript _\rho (a)\pscript _\rho (b\mid a)
\end{equation}
Now \eqref{eq22} is equivalent to $\pscript _\rho (a\circ b)=\pscript _\rho (b\circ a)$. If this holds for all $\rho\in\sscript (H)$, it follows that
$ab=ba$. Hence, Bayes' second rule does not hold, in general. When $a$ and $b$ are atomic, Theorem~\ref{thm12}(ii) characterizes the $\rho\in\lscript (H)$ for which Bayes' second rule holds.

In a similar way, for $\iscript ,\jscript\in\oscript (H)$ we have the \textit{conditional probability}
\begin{equation*} 
\pscript _\rho(\jscript\mid\iscript)=\frac{\pscript _\rho (\iscript\circ\jscript )}{\pscript _\rho (\iscript )}=
  \frac{\rmtr\sqbrac{\jscript\paren{\iscript (\rho )}}}{\rmtr\sqbrac{\iscript (\rho )}}
\end{equation*}
whenever $\pscript _\rho (\iscript )\ne 0$. As before, we write $\iscript\perp\jscript$ if $\iscript +\jscript\in\oscript (H)$. We then obtain,
$0\le\pscript _\rho (\jscript\mid\iscript )\le 1$, $\pscript _\rho (I\mid\iscript )=1$ and more generally, if $\jscript$ is a channel, then
$\pscript _\rho (\jscript\mid\iscript )=1$. Also, if $\jscript\perp\kscript$, then
$\pscript _\rho\paren{(\jscript +\kscript )\mid\iscript}=\pscript _\rho (\jscript\mid\iscript )+\pscript _\rho (\kscript\mid\iscript )$. Similar to
$\escript (H)$, Bayes' first rule says that if $\iscript _1+\iscript _2+\cdots +\iscript _n=\cscript$, where $\cscript$ is a channel, then
\begin{equation}                
\label{eq23}
\pscript (\jscript )=\sum _{i=1}^n\pscript _\rho (\iscript _i)\pscript _\rho (\jscript\mid\iscript _i)
\end{equation}
Notice that we used an arbitrary channel instead of just the trivial channel $I(\rho )=\rho$. The reason for this is that $\cscript$ corresponds to a certain event because $\pscript _\rho (\cscript )=\rmtr\sqbrac{\cscript (\rho )}=1$ for all $\rho\in\sscript (H)$. Now \eqref{eq23} is equivalent to
\begin{align}                
\label{eq24}
\rmtr\sqbrac{\jscript (\rho )}&=\pscript _\rho (\jscript )=\sum _{i=1}^n\pscript _\rho\sqbrac{\iscript _i\circ\jscript (\rho )}
   =\sum _{i=1}^n\pscript _\rho\sqbrac{\jscript\paren{\iscript _i(\rho )}}\notag\\
   &=\pscript _\rho\sqbrac{\jscript\paren{\sum _{i=1}^n\iscript _i(\rho )}}
   =\pscript _\rho\sqbrac{\jscript\paren{\cscript (\rho )}}=\rmtr\sqbrac{\jscript\paren{\cscript (\rho )}}
\end{align}
If $\cscript =I$ the trivial channel, then clearly \eqref{eq24} holds. However as we shall see in later examples, \eqref{eq24} does not hold in general.

For the purpose of examples and to better understand the structure of operations, we now consider some special types of operations. An operation
$\iscript$ is \textit{sharp} if it has the form $\iscript (\rho )=\sum a_i\rho a_i$ where $a_i$ are projections and $\iscript$ is \textit{atomic} if it is sharp and the $a_i$ are one-dimensional projections. As in the proof of Theorem~\ref{thm12}(i), if $\iscript (\rho )=\sum a_i\rho a_i$ is sharp then $a_ia_j=0$ for
$i\ne j$ so $\brac{a_i}$ are mutually orthogonal projections. We say that $\iscript\in\oscript (H)$ is \textit{Kraus} \cite{kra83} if
$\iscript (\rho )=A\rho A^*$ for some $A\in\lscript (H)$ with $A^*A\le I$ and $\iscript$ is \textit{L\"uders} \cite{lud51} if
$\iscript (\rho )=a\circ\rho =a^{1/2}\rho a^{1/2}$ for some $a\in\escript (H)$. An operation $\iscript$ is \textit{semi-trivial} if it has the form
$\iscript (\rho )=\sum\rmtr (\rho a_i)\alpha _i$ where $\alpha _i\in\sscript (H)$ and $a_i\in\escript (H)$ with $\sum a_i\le I$. Notice that this
$\iscript$ is indeed an operation because
\begin{equation*} 
\rmtr\sqbrac{\iscript (\rho )}=\rmtr\paren{\rho\sum a_i}\le\rmtr (\rho )
\end{equation*}
To be specific, we say that $\iscript$ is semi-trivial with states $\alpha _i$ and effects $a_i$, $i=1,2,\ldots ,n$. An operation $\iscript$ is \textit{trivial} if
$\iscript$ is semi-trivial with one state $\alpha$ and one effect $a$. In this case $\iscript (\rho )=\rmtr (\rho a)\alpha$ for all $\rho\in\sscript (H)$
\cite{hz12}.

If $\psi\in H$ is a unit vector we denote its corresponding projection operator by $P_\psi =\ket{\psi}\bra{\psi}$.

\begin{lem}    
\label{lem21}
An operation $\iscript$ is atomic if and only if $\iscript$ is semi-trivial with one-dimensional states $P_{\psi _i}$ and corresponding effects
$P_{\psi _i}$
\end{lem}
\begin{proof}
We have that $\iscript$ is atomic if and only if there exists one-dimensional projections $P_{\psi _i}$ such that
$\iscript (\rho )=\sum P_{\psi _i}\rho P_{\psi _i}$. But this is equivalent to
\begin{equation}                
\label{eq25}
\iscript (\rho )=\sum\ket{\psi _i}\bra{\psi _i}\rho\ket{\psi _i}\bra{\psi _i}=\!\sum\elbows{\psi _i,\rho\psi _i}\ket{\psi _i}\bra{\psi _i}
  =\sum\rmtr (\rho P_{\psi _i})P_{\psi _i}
\end{equation}
Now \eqref{eq25} is equivalent to $\iscript$ being semi-trivial with states $P_{\psi _i}$ and corresponding effects $P_{\psi _i}$.
\end{proof}

Although the Kraus operators for Kraus and L\"uders operations are obvious, this is not clear at all for semi-trivial or even trivial operations. This is treated in the next result.

\begin{thm}    
\label{thm22}
Let $\iscript (\rho )=\sum\limits _{i=1}^n\rmtr (\rho a_i)\alpha _i$ be semi-trivial and let $\alpha _i$ have the spectral representation
\begin{equation}                
\label{eq26}
\alpha _i=\sum _{j=1}^{n_i}\lambda _{ij}\ket{\phi _{ij}}\bra{\phi _{ij}}
\end{equation}
for $i=1,2,\ldots ,n$. Letting $A_{ijk}=\lambda _{ij}^{1/2}\ket{\phi _{ij}}\bra{a_i^{1/2}\phi _{ik}}$, we have that $\brac{A_{ijk}}$ is a set of Kraus operators for $\iscript$, $i=1,2,\ldots ,n$, $j,k=1,2,\ldots ,n_i$. Moreover, we have that $\iscripthat =\sum\limits _{i=1}^na_i$.
\end{thm}
\begin{proof}
Since \eqref{eq26} is a spectral representation, we conclude that $\brac{\phi _{ij}\colon j=1,2,\ldots n_i}$ is an orthonormal basis for $H$, $i=1,2,\ldots ,n$ and $\lambda _{ij}\ge 0$ with $\sum\limits _{j=1}^{n_i}\lambda _{ij}=1$. Summing over all applicable $i,j,k$, we obtain for all
$\rho\in\sscript (H)$ that
\begin{align*}
\sum _{i,j,k}A_{ijk}\rho A_{ijk}^*&=\sum _{i,j,k}\lambda _{ij}\ket{\phi _{ij}}\bra{a _i^{1/2}\phi _{ik}}\rho\ket{a_i^{1/2}\phi _{ik}}\bra{\phi _{ij}}\\
  &=\sum _{i,j,k}\lambda _{ij}\elbows{a_i^{1/2}\phi _{ik},\rho a_i^{1/2}\phi _{ik}}\ket{\phi _{ij}}\bra{\phi _{ij}}\\
  &=\sum _{i,j,k}\lambda _{ij}\elbows{\phi _{ik},a_i^{1/2}\rho a_i^{1/2}\phi _{ik}}\ket{\phi _{ij}}\bra{\phi _{ij}}\\
  &=\sum _{i,j}\lambda _{ij}\rmtr (\rho a_i)\ket{\phi _{ij}}\bra{\phi _{ij}}=\sum _i\rmtr (\rho a_i)\alpha _i=\iscript (\rho )
\end{align*}
We conclude that
\begin{align*}
\iscripthat =\sum _{i,j,k}A_{ijk}^*A_{ijk}&=\sum _{i,j,k}\lambda _{ij}\ket{a_i^{1/2}\phi _{ik}}\bra{\phi _{ij}}\,\ket{\phi _{ij}}\bra{a_i^{1/2}\phi _{ik}}\\
&=\sum _{i,j,k}\lambda _{ij}\ket{a_i^{1/2}\phi _{ik}}\bra{a_i^{1/2}\phi _{ik}}=\sum _{i,k}\ket{a_i^{1/2}\phi _{ik}}\bra{a_i^{1/2}\phi _{ik}}\\
&=\sum _{i,k}a_i^{1/2}\ket{\phi _{ik}}\bra{\phi _{ik}}a_i^{1/2}=\sum _ia_i\qedhere
\end{align*}
\end{proof}

\begin{cor}    
\label{cor23}
If $\iscript (\rho )=\rmtr (\rho a)\alpha$ is a trivial operation and $\alpha$ has spectral representation
$\alpha =\sum\limits _{i=1}^n\lambda _i\ket{\phi _i}\bra{\phi _i}$, then $A_{ij}=\lambda _i^{1/2}\ket{\phi _i}\bra{a^{1/2}\phi _j}$, $i,j=1.2.\ldots ,n$, gives a set of Kraus operators for $\iscript$. Moreover, we have that $\iscripthat =a$.
\end{cor}

It can be shown directly that $\iscripthat =\sum a_i$ in Theorem~\ref{thm22}. Indeed, we have that
\begin{equation*}
\rmtr\sqbrac{\iscript (\rho)}=\sum\rmtr (\rho a_i)=\rmtr\sqbrac{\rho\sum a_i}
\end{equation*}
Moreover, Corollary~\ref{cor23} shows that the trivial operation $\iscript (\rho )=\rmtr (\rho a)\alpha$ measures the effect $a$. Letting $\alpha$ vary, we obtain an infinite number of different operations that measure $a$. Also, the L\"uders operation $\lscript ^a(\rho )=a^{1/2}\rho a^{1/2}$ measures $a$ because
\begin{equation*}
\rmtr\sqbrac{\lscript ^a(\rho )}=\rmtr (a^{1/2}\rho a^{1/2})=\rmtr (\rho a)
\end{equation*}
for all $\rho\in\sscript (H)$. We now employ these special operations to show that Bayes' rules do not hold for operations.

\begin{exam}{1}  
Let $0<a<I$, $\alpha\in\sscript (H)$ and let $\iscript _1=\rmtr (\rho a)\alpha$ be a trivial operation. Also, let $\iscript _2$ be the trivial operation
\begin{equation*}
\iscript _2(\rho )=\rmtr (\rho a')\alpha =\sqbrac{1-\rmtr (\rho\alpha )}\alpha
\end{equation*}
Then $\cscript (\rho )=(\iscript _1+\iscript _2)(\rho )=\alpha$ is a constant channel. Now \eqref{eq24} becomes
\begin{equation}                
\label{eq27}
\rmtr\sqbrac{\jscript (\rho )}=\rmtr\sqbrac{\jscript\paren{\cscript (\rho )}}=\rmtr\sqbrac{\jscript (\alpha )}
\end{equation}
Equation~\eqref{eq27} does not hold, in general. For example, letting $\jscript =\iscript _1$, \eqref{eq27} becomes
\begin{equation}                
\label{eq28}
\rmtr (\rho a)=\rmtr\sqbrac{\iscript _1(\rho )}=\rmtr\sqbrac{\iscript _1(\alpha )}=\rmtr (\alpha a)
\end{equation}
If $\rho\ne\alpha$, then there exists an $a\in\escript (H)$ for which \eqref{eq28} fails.\hfill\qedsymbol
\end{exam}

We now give another example for which \eqref{eq24} fails.

\begin{exam}{2}  
Let $a$ be a projection and define sharp L\"uders operations $\iscript _1(\rho )=a\rho a$, $\iscript _2(\rho )=a'\rho a'$. Then
$\cscript =\iscript _1+\iscript _2$ is a channel because for all $\rho\in\sscript (H)$ we obtain
\begin{equation*}
\rmtr\sqbrac{\cscript (\rho )}=\rmtr (a\rho a+a'\rho a')=\rmtr (\rho a+\rho a')=\rmtr (\rho )=1
\end{equation*}
Let $b\in\escript (H)$ and define the L\"uders operation $\jscript (\rho )=b\circ\rho =b^{1/2}\rho b^{1/2}$. Then \eqref{eq24} becomes
\begin{align}               
\label{eq29}
\rmtr (\rho b)&=\rmtr\sqbrac{\jscript (\rho )}=\rmtr\sqbrac{\jscript\paren{\cscript (\rho )}}
  =\rmtr\sqbrac{\jscript (a\rho a)}+\rmtr\sqbrac{\jscript (a'\rho a')}\notag\\
  &=\rmtr (b^{1/2}a\rho ab^{1/2})+(b^{1/2}a'\rho a'b^{1/2})=\rmtr (\rho aba)+\rmtr (\rho a'ba')\notag\\
  &=\rmtr\sqbrac{\rho (aba+a'ba')}
\end{align}
Now \eqref{eq29} holds for all $\rho\in\sscript (H)$ if and only if $b=aba+a'ba'$. Hence, $ab=aba=ba$. Thus, if $a$ and $b$ do not commute, then \eqref{eq29} does not hold for all $\rho\in\sscript (H)$.\hfill\qedsymbol
\end{exam}

As for effects, Bayes' second rule for operations becomes $\pscript _\rho (\iscript\circ\jscript )=\pscript _\rho (\jscript\circ\iscript )$. This is equivalent to
\begin{equation}                
\label{eq210}
\rmtr\sqbrac{\jscript\paren{\iscript (\rho )}}=\rmtr\sqbrac{\iscript\paren{\jscript (\rho )}}
\end{equation}

\begin{exam}{3}  
Let $\iscript (\rho )=\rmtr (\rho a)\alpha$ and $\jscript (\rho )=\rmtr (\rho a)\beta$ be trivial operations. Then 
\begin{equation*}
\rmtr\sqbrac{\jscript\paren{\iscript (\rho )}}=\rmtr (\rho a)\rmtr\sqbrac{\jscript (\alpha )}=\rmtr (\rho a)\rmtr (\alpha a)
\end{equation*}
and similarly, $\rmtr\sqbrac{\iscript\paren{\jscript (\rho )}}=\rmtr (\rho a)\rmtr (\beta a)$. These are not the same so \eqref{eq210} fails, in general. As another example, let $\iscript (\rho )=a^{1/2}\rho a^{1/2}$, $\jscript (\rho )=b^{1/2}\rho b^{1/2}$ be L\"uders operations. We then obtain
\begin{equation*}
\rmtr\sqbrac{\jscript\paren{\iscript (\rho )}}=\rmtr (b^{1/2}a^{1/2}\rho a^{1/2}b^{1/2})=\rmtr (\rho a\circ b)]
\end{equation*}
and similarly, $\rmtr\sqbrac{\iscript\paren{\jscript (\rho )}}=\rmtr (\rho b\circ a)$. These are equal if and only if $ab=ba$ so again\eqref{eq210} fails, in general.\hfill\qedsymbol
\end{exam}

We close this section by considering complements of operations. For $\iscript ,\jscript\in\oscript (H)$, we write $\iscript\le\jscript$ if
$\iscript (\rho )\le\jscript (\rho )$ for all $\rho\in\sscript (H)$. If $\iscript\le\cscript$ where $\cscript$ is a channel, we call
$\iscript ^\cscript =\cscript -\iscript$ the $\cscript$-\textit{complement of} $\iscript$. Then $\iscript ^\cscript$ is the unique operation satisfying
$\iscript +\iscript ^\cscript =\cscript$. An operation other than $I$ can have many complements. For example, any channel is a complement of $0$. We frequently say, for short, that $\jscript$ is a complement of $\iscript$ instead of $\jscript$ is a $\cscript$-complement of $\iscript$.

\begin{exam}{4}  
Let $\iscript\in\oscript (H)$ with Kraus decomposition $\iscript (\rho )=\sum A_i\rho A_i^*$. Letting $b=\sum A_i^*A_i$ we have that $b\in\escript (H)$. Then the L\"uders operation $\jscript (\rho )=(I-b)^{1/2}\rho (I-b)^{1/2}$ is a complement of $\iscript$ because
\begin{equation*}
(\iscript +\jscript )(\rho )=\iscript (\rho )+\jscript (\rho )=\sum A_i\rho A_i^*+(I-b)^{1/2}\rho (I-b)^{1/2}
\end{equation*}
and since $\sum A_i^*A_i+I-b=I$, we conclude that $\iscript +\jscript$ is a channel. This shows that any operation has a unique L\"uders operation complement.\hfill\qedsymbol
\end{exam}

\begin{exam}{5}  
Let $\iscript (\rho )=a\circ\rho =a^{1/2}\rho a^{1/2}$ be a L\"uders operation with $a\in\escript (H)$. Then $\jscript (\rho )=a'\circ\rho$ is a complement of $\iscript$ because $\iscript (\rho )+\jscript (\rho )=a\circ \rho +a'\circ\rho$ is a channel. Moreover, if $\iscript (\rho )=\rmtr (\rho a)\alpha$ is trivial, then $\jscript (\rho )=\rmtr (\rho a')\alpha$ is a complement of $\iscript$ because
\begin{equation*}
(\iscript +\jscript )(\rho )=\iscript (\rho )+\jscript (\rho )=\rmtr (\rho a)\alpha +\rmtr (\rho a')\alpha =\rmtr (\rho )\alpha =\alpha
\end{equation*}
is a channel.\hfill\qedsymbol
\end{exam}

\begin{thm}    
\label{thm24}
{\rm{(i)}}\enspace $\jscript$ is a complement of $\iscript$ if and only if $\jscripthat =(\,\iscript\,)'$.
{\rm{(ii)}}\enspace If $\iscript$ is sharp, then $\iscript\wedge\iscript ^\cscript =0$ for some channel $\cscript$.
\end{thm}
\begin{proof}
(i)\enspace Suppose that $\jscript$ is a $\cscript$-complement of $\iscript$ so $\jscript =\iscript ^\cscript =\cscript -\iscript$. Then
\begin{equation*}
\jscripthat =\cscripthat -\iscripthat =I-\iscripthat =(\,\iscripthat\,)'
\end{equation*}
Conversely, if $\jscripthat =(\,\iscripthat\,)'$ then
\begin{equation*}
\rmtr\sqbrac{\jscript (\rho )+\iscript (\rho )}=\rmtr (\rho\jscript\,)+\rmtr (\rho\iscripthat\,)=\rmtr (\rho )=1
\end{equation*}
for all $\rho\in\sscript (H)$. Hence, $\jscript +\iscript$ is a channel so $\jscript$ is a complement of $\iscript$.
(ii)\enspace. Let $\iscript (\rho )=\sum a_i\rho a_i$ where $a_i$ are projections and let $\jscript (\rho )=b'\rho b'$ where $b=\sum a_i$. Then letting
\begin{equation*}
\cscript (\rho )=\iscript (\rho )+\jscript (\rho )=\sum a_i\rho a_i+b'\rho b'
\end{equation*}
we have that
\begin{equation*}
\rmtr\sqbrac{\cscript (\rho )}=\rmtr\sqbrac{\rho\paren{\sum a_i+b'}}=\rmtr (\rho )=1
\end{equation*}
for every $\rho\in\sscript (H)$ so $\cscript$ is a channel. Hence, $\jscript =\iscript ^\cscript$. To show that $\iscript\wedge\jscript =0$, let
$\kscript\le\iscript ,\jscript$. Then $\kscripthat\le\iscripthat =\sum a_i$ and
\begin{equation*}
\kscripthat\le\jscripthat =b'=I-\sum a_i
\end{equation*}
It follows that $\kscripthat =0$ so $\kscript =0$.
\end{proof}

\section{Observables and Instruments}  
We now extend our previous work to observables and instruments. An \textit{observable} is a finite set
$A=\brac{a_x\colon x\in\Omega _A}\subseteq\escript (H)$ satisfying $\sum\limits _{x\in\Omega _A}a_x=I$ \cite{bgl95,fhl18,hz12,nc00}. We call
$\Omega _A$ the \textit{outcome set} and $x\in\Omega _A$ is an \textit{outcome} for $A$. We think of $A$ as an experiment with possible outcomes $x\in\Omega _A$ and $a_x$ is the effect that is $yes$ when $A$ has outcome $x$. The \textit{probability} that $A$ has outcome $x$ when the system is in state $\rho\in\sscript (H)$ is $\pscript _\rho (x)=\rmtr (\rho a_x)$ and we call $\Phi _\rho ^A(x)=\pscript _\rho (x)$ the \textit{distribution} of $A$. If
$\Delta\subseteq\Omega _A$, we define the probability of $\Delta$ in the state $\rho$ by
\begin{equation*} 
\Phi _\rho ^A(\Delta )=\pscript _\rho (\Delta )=\sum\brac{\pscript _\rho (x)\colon x\in\Delta}=\sum\brac{\rmtr (\rho a_x)\colon x\in\Delta}
\end{equation*}
We see that $\Delta\mapsto\sum _{x\in\Delta}A_x$ is an effect-valued measure on $2^{\Omega _A}$. We denote the set of observables on $H$ by
$\rmob (H)$.

If $A=\brac{a_x\colon x\in\Omega _A}$ and $B=\brac{b_y\colon y\in\Omega _B}$ are observables, we define their \textit{sequential product} \cite{gud120,gud220,gud320}
\begin{equation*} 
(A\circ B)_{(x,y)}=\brac{a_x\circ b_y\colon (x,y)\in\Omega _A\times\Omega _B}
\end{equation*}
with outcome set $\Omega _{A\circ B}=\Omega _A\times\Omega _B$. Notice that $A\circ B$ is indeed an observable because
\begin{equation*}
\sum _{(x,y)\in\Omega _{A\circ B}}(A\circ B)_{(x,y)}=\sum _{x,y}(a_x\circ b_y)=\sum _xa_x\circ\sum _yb_y=\sum _xa_x=I
\end{equation*}
We also have the observable $B$ \textit{conditioned by} the observable $A$ defined as \cite{gud120}
\begin{equation*}
(B\mid A)_y=\sum _{x\in\Omega _A}(A\circ B)_{(x,y)}=\sum _{x\in\Omega _A}(a_x\circ b_y)
\end{equation*}
where $\Omega _{(B\mid A)}=\Omega _B$. Again, $(B\mid A)$ is an observable because
\begin{equation*}
\sum _{y\in\Omega _B}(B\mid A)_y=\sum _{x,y}(a_x\circ b_y)=I
\end{equation*}
Just as for effects, Bayes' rule
\begin{equation*}
\pscript _\rho (b_y)=\sum _x\pscript _\rho (a_x)\pscript _\rho (b_y\mid a_x)
\end{equation*}
does not hold. However, we do have the result
\begin{equation*}
\pscript _\rho\sqbrac{(B\mid A)_y}=\sum _x\pscript _\rho (a_x)\pscript _\rho (b_y\mid a_x)
\end{equation*}
Indeed, for all $\rho\in\sscript (H)$ we obtain
\begin{align*}
\pscript _\rho\sqbrac{(B\mid A)_y}&=\rmtr\sqbrac{\rho (B\mid A)_y}=\rmtr\sqbrac{\rho\sum _x(a_x\circ b_y)}=\sum _x\rmtr(\rho a_x\circ b_y)\\
   &=\sum _x\rmtr (a_x^{1/2}\rho a_x^{1/2}b_y)=\sum _x\rmtr (\rho a_x)\pscript _\rho (b_y\mid a_x)\\
   &=\sum _x\pscript _\rho (a_x)\pscript _\rho (b_y\mid a_x)
\end{align*}
Also, notice that $y\mapsto\pscript _\rho (b_y\mid a_x)$ is additive so it is a real-valued measure.

An \textit{instrument} \cite{bgl95,hz12,nc00} is a finite set $\iscript =\brac{\iscript _x\colon x\in\Omega _\iscript}\subseteq\oscript (H)$ satisfying
$\iscriptbar =\sum _{x\in\Omega _\iscript}\iscript _x$ is a channel. We call $\Omega _\iscript$ the \textit{outcome set} and $x\in\Omega _\iscript$ is an \textit{outcome} for $\iscript$. The \textit{probability} that $\iscript$ has outcome $x$ when the system is in state $\rho\in\sscript (H)$ is
$\pscript _\rho (x)=\rmtr\sqbrac{\iscript _x(\rho )}$ and we call $\Phi _\rho ^\iscript (x)=\pscript _\rho (x)$ the \textit{distribution} of $\iscript$. If
$\Delta\subseteq\Omega _\iscript$, we define the probability of $\Delta$ in the state $\rho$ by
\begin{equation*}
\Phi _\rho ^\iscript (\Delta )=\pscript _\rho (\Delta )=\sum\brac{\pscript _\rho (x)\colon x\in\Delta}
   =\sum\brac{\rmtr\sqbrac{\iscript _x(\rho )}\colon x\in\Delta}
\end{equation*}
Then $\Delta\mapsto\sum _{x\in\Delta}\iscript _x$ is an operation-valued measure on $2^{\Omega _\iscript}$. We denote the set of instruments on $H$ by $\rmin (H)$. We say that $\iscript\in\rmin (H)$ \textit{measures} $A\in\rmob (H)$ if $\Omega _A=\Omega _\iscript$ and
$\Phi _\rho ^A(x)=\Phi _\rho ^\iscript (x)$ for all $x\in\Omega _A$, $\rho\in\sscript (H)$. Since this is equivalent to $\iscripthat _x=a_x$, we have that
$\iscript$ measures a unique $\iscripthat\in\rmob (H)$ given by $\iscripthat =\brac{\iscripthat _x\colon x\in\Omega _\iscript}$. We think of $\iscript$ as an apparatus that is employed to measure the observable $\iscripthat$. Although $\iscript$ measures the unique $\iscripthat\in\rmob (H)$, as we shall see, an observable is measured by many instruments.

If $\iscript =\brac{\iscript _x\colon x\in\Omega _\iscript}$ and $\jscript =\brac{\jscript _y\colon y\in\Omega _\jscript}$ are instruments, we define their 
\textit{sequential product} \cite{gud120,gud220,gud320}
\begin{equation*}
(\iscript\circ\jscript )_{(x,y)}=\brac{\iscript _x\circ\jscript _y\colon (x,y)\in\Omega _\iscript\times\Omega _\jscript}
\end{equation*}
with outcome set $\Omega _{\iscript\circ\jscript}=\Omega _\iscript\times\Omega _\jscript$. We see that $\iscript\circ\jscript$ is indeed an instrument because
\begin{equation*}
\sum _{(x,y)\in\Omega _{\iscript\circ\jscript}}(\iscript\circ\jscript )_{(x,y)}=\sum _{x,y}\iscript _x\circ\jscript _y
   =\sum _x\iscript _x\circ\sum _y\jscript _y=\iscriptbar\circ\jscriptbar
\end{equation*}
which is a channel. We also have the instrument $\jscript$ \textit{conditioned by} $\iscript$ defined as
\begin{equation*}
(\jscript\mid\iscript )_y(\rho )=\sum _x(\iscript _x\circ\jscript _y)(\rho )=\sum _x\jscript _y\sqbrac{\iscript _x(\rho )}
=\jscript _y\sqbrac{\sum _x\iscript _x(\rho )}=\jscript _y\sqbrac{\,\iscriptbar (\rho )}
\end{equation*}
We have that $(\jscript\mid\iscript )$ is indeed an instrument with outcome space $\Omega _\jscript$ because
\begin{equation*}
\sum _y(\jscript\mid\iscript )_y=\sum _y\jscript _y\sqbrac{\,\iscriptbar (\rho )}=\jscriptbar\sqbrac{\,\iscriptbar (\rho )}
\end{equation*}
which is a channel. As with observables, we have that
\begin{equation*}
\pscript _\rho\sqbrac{(\jscript\mid\iscript )_y}=\sum _x\pscript _\rho (\iscript _x)\pscript _\rho (\jscript _y\mid\iscript _x)
\end{equation*}
because
\begin{align*}
\pscript _\rho\sqbrac{(\jscript\mid\iscript )_y}&=\rmtr\sqbrac{(\jscript\mid\iscript )_y}=\rmtr\brac{\sum _x\jscript _y\sqbrac{\iscript _x(\rho )}}\\
   &=\sum _x\rmtr\brac{\sqbrac{\iscript _x(\rho )}\pscript _\rho (\jscript _y\mid\iscript _x)}
   =\sum _x\pscript _\rho (\iscript _x)\pscript _\rho (\jscript _y\mid\iscript _x)
\end{align*}
Also notice that $y\mapsto\pscript _\rho (\jscript _y\mid\iscript _x)$ is additive and a probability measure because
\begin{align*}
\sum _y\pscript _\rho (\jscript _y\mid\iscript _x)&=\frac{1}{\rmtr\sqbrac{\iscript _x(\rho )}}\rmtr\brac{\sum _y\jscript _y\sqbrac{\iscript _x(\rho )}}
   =\frac{1}{\rmtr\sqbrac{\iscript _x(\rho )}}\rmtr\sqbrac{\,\jscriptbar\paren{\iscript _x(\rho )}}\\
   &=\rmtr\sqbrac{\,\jscriptbar\paren{\frac{\iscript _x(\rho )}{\rmtr\paren{\iscript _x(\rho )}}}}=1
\end{align*}
As with operations, Bayes' rules do not hold for instruments.

We now consider various types of instruments. A general instrument $\iscript =\brac{\iscript _x\colon x\in\Omega _\iscript}$ has a Kraus decomposition
$\iscript _x(\rho )=\sum\limits _{i=1}^{n_x}A_i^x\rho (A_i^x)^*$ with $\iscripthat _x=\sum\limits _{i=1}^{n_x}(A_i^x)^*A_i^x\le I$ and
$\iscriptbar (\rho )=\sum\limits _x\sum\limits _{i=1}^{n_x}A_i^x\rho (A_i^x)^*$ with $\sum\limits _x\sum\limits _{i=1}^{n_x}(A_i^x)^*A_i^x=I$. We say that
$\iscript$ is \textit{sharp} if $A_i^x$ are projections for all $i,x$. We then have that $\sum\limits _{x,i}A_i^x=I$ and it follows that $A_i^xA_j^y=0$ if
$(x,i)\ne (y,j)$. We say that $\iscript$ is \textit{atomic} if $A_i^x$ are one-dimensional projections for all $i,x$. An instrument $\iscript$ is
\textit{Kraus} \cite{kra83} if it has Kraus decompositions $\iscript _x(\rho )=A_x\rho A_x^*$ for all $x\in\Omega _\iscript$ in which case
$\sum\limits _xA_x^*A_x=I$ and $\iscripthat =\brac{A_x^*A_x\colon x\in\Omega _\iscript}$. An instrument $\iscript$ is \textit{L\"uders} \cite{lud51} if it is Kraus and has decomposition
\begin{equation*}
\iscript _x(\rho )=a_x\circ\rho =a_x^{1/2}\rho a_x^{1/2}
\end{equation*}
for every $x\in\Omega _\iscript$ where $a_x\in\escript (H)$. In this case $\sum\limits _xa_x=I$ and we obtain the observable
$\iscripthat =\brac{a_x\colon x\in\Omega _\iscript}$. We then use the notation $A=\iscripthat$ and write $\iscript =\lscript ^A$. It follows that
$(\lscript ^A)^\wedge =A$. An instrument $\iscript$ is \textit{trivial} if there exists an observable $A=\brac{a_x\colon x\in\Omega _\iscript}$ and
$\alpha\in\sscript (H)$ such that $\iscript _x(\rho )=\rmtr (\rho a_x)\alpha$ for all $x\in\Omega _\iscript$. In this case $\iscript _x=a_x$ so $\iscript$ measures $A$. We then say that $\iscript$ is trivial with observable $A$ and state $\alpha$. More generally, we say that $\iscript$ is
\textit{semi-trivial} with observable $A=\brac{a_x\colon x\in\Omega _\iscript}$ and states $\alpha _x$ if $\iscript _x(\rho )=\rmtr (\rho a_x)\alpha _x$. We again have that $\iscripthat _x=a_x$ so $\iscript$ measures $A$. These last three types illustrate that an observable can be measured by many different instruments

\begin{thm}    
\label{thm31}
{\rm{(i)}}\enspace For any $\iscript ,\jscript\in\rmin (H)$ we have
\begin{equation*}
\overline{(\iscript\circ\jscript )}=\overline{(\jscript\mid\iscript)}=\iscriptbar\circ\jscriptbar
\end{equation*}
{\rm{(ii)}}\enspace If $A\in\rmob (H)$ and $\iscript\in\rmin (H)$, then
\begin{equation*}
(\lscript ^A\circ\iscript )^\wedge =(\lscript ^A)^\wedge\circ\iscripthat =A\circ\iscripthat
\end{equation*}
{\rm{(iii)}}\enspace If $A,B\in\rmob (H)$, then
\begin{equation*}
(\lscript ^A\circ\lscript ^B)^\wedge =(\lscript ^A)^\wedge\circ (\lscript ^B)^\wedge =A\circ B
\end{equation*}
\end{thm}
\begin{proof}
(i)\enspace We have that
\begin{equation*}
\overline{(\iscript\circ\jscript )}=\sum _{x,y}(\iscript\circ\jscript )_{(x,y)}=\sum _{x,y}\iscript _x\circ\jscript _y=\sum _x\iscript _x\circ\sum _y\jscript _y
   =\iscriptbar\circ\jscriptbar
\end{equation*}
Moreover,
\begin{equation*}
\overline{(\jscript\mid\iscript )}=\sum _y(\jscript\mid\iscript )_y=\sum _y\sum _x\iscript _x\circ\jscript _y=\iscriptbar\circ\jscriptbar
\end{equation*}
(ii)\enspace Letting $A=\brac{a_x\colon x\in\Omega _A}$ we obtain
\begin{align*}
\rmtr\sqbrac{\rho (\lscript ^A\circ\iscript )_{(x,y)}^\wedge}&=\rmtr\sqbrac{(\lscript ^A\circ\iscript )_{(x,y)}(\rho )}
   =\rmtr\sqbrac{(\lscript _x^A\circ\iscript _y)(\rho )}\\
   &=\rmtr\sqbrac{\iscript _y(a_x^{1/2}\rho a_x^{1/2})}=\rmtr\sqbrac{a_x^{1/2}\rho a_x^{1/2}\iscripthat _y}\\
   &=\rmtr (\rho a_x^{1/2}\iscripthat _ya_x^{1/2})=\rmtr (\rho a_x\circ\iscripthat _y)=\rmtr\sqbrac{\rho (A\circ\iscripthat\,)_{(x,y)}}
\end{align*}
for all $\rho\in\sscript (H)$. Hence, $(\lscript ^A\circ\iscript )_{(x,y)}^\wedge =(A\circ\iscripthat\,)_{(x,y)}$ for all
$(x,y)\in\Omega _A\times\Omega _\iscript$. We conclude that $(\lscript ^A\circ\iscript )^\wedge=A\circ\iscripthat =(\lscript ^A)^\wedge\circ\iscripthat$.
(iii) follows from (ii).
\end{proof}

\begin{exam}{6}  
Unlike Theorem~\ref{thm31}(ii) we show that $(\iscript\circ\lscript ^A)^\wedge\ne\iscripthat\circ A$, in general. Let $\iscript _x(\rho )=\rmtr (\rho b_x)\alpha$ be a trivial instrument. Letting $A=\brac{a_y\colon y\in\Omega _A}$ we have
\begin{equation*}
(\iscripthat\circ A)_{(x,y)}=\iscripthat _x\circ a_y=b_x\circ a_y=b_x^{1/2}a_yb_x^{1/2}
\end{equation*}
On the other hand, since
\begin{align*}
\rmtr\sqbrac{\rho (\iscript\circ\lscript ^A)_{(x,y)}^\wedge}&=\rmtr\sqbrac{(\iscript\circ\lscript ^A)_{(x,y)}(\rho )}
   =\rmtr\sqbrac{(\iscript _x\circ\lscript _y^A)(\rho )}\\
   &=\rmtr\sqbrac{a_y^{1/2}\iscript _x(\rho )a_y^{1/2}}=\rmtr (\rho b_x)\rmtr (a_y^{1/2}\alpha a_y^{1/2})\\
   &=\rmtr (\rho b_x)\rmtr (\alpha a_y)=\rmtr\sqbrac{\rho\,\rmtr (\alpha a_y)b_x}
\end{align*}
we have that $(\iscript\circ\lscript ^A)_{(x,y)}^\wedge =\rmtr (\alpha a_y)b_x$ it is clear that $b_x^{1/2}a_yb_x^{1/2}\ne\rmtr (\alpha a_y)b_x$,
in general.\hfill\qedsymbol
\end{exam}

\begin{exam}{7}  
Let $A=\brac{a_x\colon x\in\Omega _A}$, $B=\brac{b_y\colon y\in\Omega _B}$ be observable and let $\iscript _x(\rho )=\rmtr (\rho a_x)\alpha$,
$\jscript _y(\rho )=\rmtr (\rho b_y)\beta$ be trivial instruments. We show that $(\iscript\circ\jscript )^\wedge\ne\iscripthat\circ\jscripthat$, in general. We have that
\begin{equation*}
\rmtr\sqbrac{\rho (\iscript\circ\jscript )_{(x,y)}^\wedge}=\rmtr\sqbrac{(\iscript\circ\jscript )_{(x,y)}(\rho )}=\rmtr\sqbrac{\jscript _y(\iscript _x(\rho ))}
   =\rmtr (\rho a_x)\rmtr (\alpha b_y)
\end{equation*}
However, $\rmtr (\rho\iscripthat _x\circ\jscripthat _y)=\rmtr (\rho a_xb_y)$ and these do not agree, in general. For example, if $a_x=\ket{\phi}\bra{\phi}$,
$b_y=\ket{\psi}\bra{\psi}$ where $\ab{\elbows{\phi ,\psi}}^2\ne\elbows{\psi ,\alpha\psi}$ then
\begin{align*}
\rmtr\sqbrac{\rho (\iscript\circ\jscript )_{(x,y)}^\wedge}&=\elbows{\phi\rho\phi}\elbows{\psi ,\alpha\psi}\ne\ab{\elbows{\phi ,\psi}}^2\elbows{\phi ,\rho\phi}\\
    &=\ab{\elbows{\phi ,\psi}}^2\rmtr\paren{\rho\ket{\phi}\bra{\phi}}=\rmtr\paren{\rho\ket{\phi}\bra{\phi}\,\ket{\psi}\bra{\psi}\,\ket{\phi}\bra{\phi}}\\
    &=\rmtr(\rho a_x\circ b_y)
\end{align*}
Hence, $\iscripthat\circ\jscripthat =A\circ B\ne (\iscript\circ\jscript )^\wedge$.\hfill\qedsymbol
\end{exam}

\begin{exam}{8}  
We give another example in which $(\iscript\circ\jscript )^\wedge\ne\iscripthat\circ\jscripthat$. Let $\iscript _x(\rho )=A_x\rho A_x^*$,
$\jscript _y(\rho )=B_y\rho B_y^*$ be Kraus instruments. Then
\begin{equation*}
\rmtr\sqbrac{(\iscript\circ\jscript )_{(x,y)}(\rho )}=\rmtr\sqbrac{\jscript _y\paren{\iscript _x(\rho )}}=\rmtr\sqbrac{B_y(A_x\rho A_x^*)B_y^*}
   =\rmtr (\rho A_x^*B_y^*B_yA_x)
\end{equation*}
Hence, $(\iscript\circ\jscript )_{(x,y)}^\wedge =A_x^*B_y^*B_yA_x$. On the other hand,
\begin{equation*}
(\,\iscripthat\circ\jscripthat\,)_{(x,y)}=\iscripthat _x\circ\jscripthat _y=(A_x^*A_x)\circ (B_y^*B_y)=(A_x^*A_x)^{1/2}B_y^*B_y(A_x^*A_x)^{1/2}
\end{equation*}
and these are not equal, in general.\hfill\qedsymbol
\end{exam}

\begin{lem}    
\label{lem32}
{\rm{(i)}}\enspace If $\iscript$ and $\jscript$ are semi-trivial with states $\alpha _x$, observable $A=\brac{a_x}$ and states $\beta _y$, observable
$B=\brac{b_y}$, respectively, then $\iscript\circ\jscript$ is semi-trivial with states $\beta _y$ and observable $\brac{\rmtr (\alpha _xb_y)a_x}$. Also,
$(\jscript\mid\iscript )$ is semi-trivial with states $\beta _y$ and observable $\brac{\sum\limits _x\rmtr (\alpha _x\beta _y)a_x}$.
{\rm{(ii)}}\enspace If $\iscript$ and $\jscript$ are trivial with state $\alpha$, observable $\brac{a_x}$ and state $\beta$, observable $\brac{b_y}$, respectively, the $\iscript\circ\jscript$ is trivial with state $\beta$ and observable $\brac{\rmtr (\alpha b_y)a_x}$. Also, $(\jscript\mid\iscript )$ is trivial with state $\beta$ and observable $\brac{\rmtr (\alpha b_y)I}$.
\end{lem}
\begin{proof}
(i)\enspace Since
\begin{align*}
(\iscript\circ\jscript )_{(x,y)}(\rho )&=\jscript _y\sqbrac{\iscript _x(\rho )}=\jscript _y\sqbrac{\rmtr (\rho a_x)\alpha _x}
   =\rmtr (\rho a_x)\rmtr (\alpha _xb_y)\beta _y\\
   &=\rmtr\sqbrac{\rho\,\rmtr (\alpha _xb_y)a_x}\beta _y
\end{align*}
We conclude that $\iscript\circ\jscript$ is semi-trivial with states $\beta _y$ and observable $\rmtr (\alpha _xb_y)a_x$. Moreover, since
\begin{equation*}
(\jscript\mid\iscript )_y(\rho )=\jscript _y\sqbrac{\,\iscriptbar\,(\rho )}=\jscript _y\sqbrac{\sum _x\rmtr (\rho a_x)\alpha _x}
  =\rmtr\sqbrac{\rho\sum _x\rmtr (\alpha _xb_y)a_x}\beta _y
\end{equation*}
we conclude that $(\jscript\mid\iscript )$ is semi-trivial with states $\beta _y$ and observable $\brac{\sum\limits _x\rmtr (\alpha _x\beta _y)a_x}$.
(ii) follows from (i).
\end{proof}

We have seen that the sequential product of trivial (semi-trivial) instruments is trivial (semi-trivial). Also, the sequential product of two Kraus instruments $\iscript _x(\rho )=A_x\rho A_x^*$ and $\jscript _y(\rho )=B_y\rho B_y^*$ is Kraus because
\begin{equation*}
\iscript _x\circ\jscript _y(\rho )=\jscript _y\paren{\iscript _x(\rho )}=B_yA_x\rho A_x^*B_y^*=B_yA_x\rho (B_yA_x)^*
\end{equation*}
so the Kraus operator for $\iscript _x\circ\jscript _y$ is $B_yA_x$. However, the sequential  product of two L\"uders instruments
$\iscript _x(\rho )=a_x^{1/2}\rho a_x^{1/2}$, $\jscript _y(\rho )=b_y^{1/2}\rho b_y^{1/2}$ given by
\begin{equation*}
\iscript\circ\jscript _y=a_x^{1/2}b_y^{1/2}\rho b_y^{1/2}a_x^{1/2}
\end{equation*}
need not be L\"uders.

\begin{lem}    
\label{lem33}
If $A=\brac{a_x\colon x\in\Omega _A}$ and $B=\brac{b_y\colon y\in\Omega _B}$ are observables, then
\begin{equation*}
(\lscript ^B\mid\lscript ^A)^\wedge=\paren{(\lscript ^B)^\wedge\mid (\lscript ^A)^\wedge}=(B\mid A)
\end{equation*}
\end{lem}
\begin{proof}
For all $\rho\in\sscript (H)$ we obtain
\begin{align*}
\rmtr\sqbrac{\rho (\lscript ^B\mid\lscript ^A)_y^\wedge}&=\rmtr\sqbrac{(\lscript ^B\mid\lscript ^A)_y(\rho )}
   =\rmtr\sqbrac{\lscript _y^B\paren{\,{\overline{\lscript ^A}}(\rho )}}=\rmtr\sqbrac{\lscript _y^B\paren{\sum _xa_x^{1/2}\rho a_x^{1/2}}}\\
   &=\sum _x\rmtr\sqbrac{\lscript _y^B(a_x^{1/2}\rho a_x^{1/2})}=\sum _x\rmtr (b_y^{1/2}a_x^{1/2}\rho a_x^{1/2}b_y^{1/2})\\
   &=\sum _x\rmtr (\rho a_x\circ b_y)=\rmtr\sqbrac{\sum _x(a_x\circ b_y)}=\rmtr\sqbrac{\rho (B\mid A)_y}
\end{align*}
It follows that $(\lscript ^B\mid\lscript ^A)_y^\wedge =(B\mid A)_y$ as hence, the result holds.
\end{proof}

\begin{exam}{9}  
This example shows that $(\jscript\mid\iscript )^\wedge\ne (\,\jscripthat\mid\iscripthat\,)$, in general. Letting $\iscript _x(\rho )=\rmtr (\rho a_x)\alpha$,
$\jscript _y(\rho )=\rmtr (\rho b_y)\beta$ be trivial instruments, we obtain
\begin{equation*}
\rmtr\sqbrac{\rho (\jscript\mid\iscript )_y^\wedge}=\rmtr\sqbrac{\jscript _y\paren{\,\iscriptbar (\rho )}}
   =\rmtr\sqbrac{\rmtr\paren{\,\iscriptbar (\rho )b_y}\beta}=\rmtr\sqbrac{\,\iscriptbar (\rho )b_y}=\rmtr (\alpha b_y)
\end{equation*}
On the other hand,
\begin{equation*}
\rmtr\sqbrac{\rho (\,\jscripthat\mid\iscripthat\,)_y}=\rmtr\sqbrac{\rho (B\mid A)_y}=\rmtr\sqbrac{\rho\sum _x(a_x\circ b_y)}
\end{equation*}
These do not agree, in general. For example, if $\Omega _\iscript =\brac{x}$ so that $\iscript _x(\rho )=\alpha$ for all $\rho\in\sscript (H)$, then we obtain 
\begin{equation*}
\rmtr\sqbrac{\rho (\,\jscripthat\mid\iscripthat\,)_y}=\rmtr (\rho b_y)\ne\rmtr (\alpha b_y)
\end{equation*}
in general.\hfill\qedsymbol
\end{exam}

\section{Mixed Sequential Products and Coexistence}  
We have defined sequential products of effects (observables) and of operations (instruments). We now define mixed sequential products for effects with operations and observables with instruments. If $a\in\escript (H)$ and $\iscript\in\oscript (H)$, we define $a\circ\iscript =\lscript ^a\circ\iscript$. Thus,
$a\circ\iscript\in\oscript (H)$ and $(a\circ\iscript )(\rho )=\iscript (a^{1/2}\rho a^{1/2})$ for all $\rho\in\sscript (H)$. Notice that $a\circ I=\lscript ^a$ and
$I\circ\iscript =\iscript$. Also, it is easy to check that $a\circ\iscript$ is additive and affine in the second argument but is not in the first. If
$\iscript\in\oscript (H)$ has Kraus decomposition $\iscript (\rho )=\sum B_i\rho B_i^*$ and $a\in\escript (H)$, we define $\iscript\circ a\in\escript (H)$ by $\iscript\circ a=\sum B_i^*aB_i$. This definition does not depend on the particular Kraus operators because if $\iscript$ also has the decomposition $\iscript (\rho )=\sum C_i\rho C_i^*$, then for every $\rho\in\sscript (H)$ we obtain
\begin{align*}
\rmtr\paren{\rho\sum B_i^*aB_i}&=\sum\rmtr (\rho B_i^*aB_i)=\sum\rmtr (B_i\rho B_i^*a)=\rmtr\paren{\sum B_i\rho B_i^*a}\\
   &=\rmtr\paren{\sum C_i\rho C_i^*a}=\rmtr\paren{\rho\sum C_i^*aC_ii}
\end{align*}
Hence, $\sum B_i^*aB_i=\sum C_i^*aC_i$. It is easy to check that $\iscript\circ a$ is additive and affine in both variables.

We now extend these definitions to observables and instruments. If $A=\brac{a_x\colon x\in\Omega _A}$ is an observable and $\iscript\in\rmin (H)$, we define $A\circ\iscript\in\rmin (H)$ to have outcome space $\Omega _{A\circ\iscript}=\Omega _A\times\Omega _\iscript$ and
\begin{equation*}
(A\circ\iscript )_{(x,y)}(\rho )=(a_x\circ\iscript _y)(\rho )=(\lscript ^{a_x}\circ\iscript _y)(\rho )=\iscript _y\paren{\lscript ^{a_x}(\rho )}
   =\iscript _y(a_x^{1/2}\rho a_x^{1/2})
\end{equation*}
It is then natural to define $(\iscript\mid A)\in\rmin (H)$ to have $\Omega _{(\iscript\mid A)}=\Omega _\iscript$ and
\begin{align*}
(\iscript\mid A)_y(\rho )&=\sum _x(A\circ\iscript )_{(x,y)}(\rho )=\sum _x\iscript _y(a_x^{1/2}\rho a_x^{1/2})
   =\iscript _y\paren{\sum _xa_x^{1/2}\rho a_x^{1/2}}\\
   &=\iscript _y\paren{\,{\overline{\lscript ^A}}(\rho )}=\paren{\,{\overline{\lscript ^A}}\circ\iscript _y}(\rho )
\end{align*}
We thus define $(\iscript\mid A)_y={\overline{\lscript ^A}}\circ\iscript _y$. If $\iscript\in\rmin (H)$ has Kraus decomposition
$\iscript _x(\rho )=\sum B_i^x\rho (B_i^x)^*$ and $A=\brac{a_x\colon x\in\Omega _A}\in\rmob (H)$ we define $\iscript\circ A\in\rmob (H)$ to have outcome space $\Omega _{\iscript\circ A}=\Omega _\iscript\times\Omega _A$ and
\begin{equation*} 
(\iscript\circ A)_{(x,y)}=\iscript _x\circ a_y=\sum (B_i^x)^*a_yB_i^x
\end{equation*}
It is then natural to define $(A\mid\iscript )\in\rmob (H)$ to have $\Omega _{(A\mid\iscript )}=\Omega _A$ and
\begin{equation*}
(A\mid\iscript )_y=\sum _x(\iscript\circ A)_{(x,y)}=\sum _{x,i}(B_i^x)^*a_yB_i^x=\iscriptbar\circ a_y
\end{equation*}

\begin{thm}    
\label{thm41}
{\rm{(i)}}\enspace If $a\in\escript (H)$ and $\iscript\in\oscript (H)$, then $(a\circ\iscript )^\wedge =a\circ\iscripthat$.
{\rm{(ii)}}\enspace If $A\in\rmob (H)$ and $\iscript\in\rmin (H)$, then $(A\circ\iscript )_{(x,y)}^\wedge =a_x\circ\iscripthat _y=(A\circ\iscripthat )_{(x,y)}$.
{\rm{(iii)}}\enspace If $A\in\rmob (H)$, $\iscript\in\rmin (H)$, then $(\iscript\mid A)^\wedge =(\,\iscripthat\mid A)$.
\end{thm}
\begin{proof}
(i)\enspace For all $\rho\in\sscript (H)$ we have that
\begin{align*}
\rmtr\sqbrac{\rho (a\circ\iscript )^\wedge}&=\rmtr\sqbrac{(a\circ\iscript )(\rho )}=\rmtr\sqbrac{\iscript (a^{1/2}\rho a^{1/2})}
   =\rmtr(a^{1/2}\rho a^{1/2}\iscripthat\,)\\
   &\rmtr (\rho a^{1/2}\iscripthat a^{1/2})=\rmtr (\rho a\circ\iscripthat\,)
\end{align*}
and the result follows.
(ii)\enspace Since $(A\circ\iscript )_{(x,y)}=a_x\circ\iscript _y$, by (i) we obtain
\begin{equation*}
(A\circ\iscript )_{(x,y)}^\wedge =(a_x\circ\iscript _y)^\wedge=a_x\circ\iscripthat _y=(A\circ\iscripthat\,)_{(x,y)}
\end{equation*}
(iii)\enspace Applying Theorem~\ref{thm31}(ii) we have for all $\rho\in\sscript (H)$ that
\begin{align*}
\rmtr\sqbrac{\rho (\iscript\mid A)_y^\wedge}&=\rmtr\sqbrac{(\iscript\mid A)_y(\rho )}=\rmtr\sqbrac{(\,{\overline{\lscript ^A}}\circ\iscript _y)(\rho )}
   =\sqbrac{\rho (\lscript ^A\circ\jscript _y)^\wedge}\\
   &=\rmtr\sqbrac{\rho\paren{\sum _x\lscript ^{a_x}\circ\iscript _y}^\wedge}=\sum _x\rmtr\sqbrac{\rho (\lscript ^{a_x}\circ\iscript _y)^\wedge}\\
   &=\sum _x\rmtr\sqbrac{\rho (a_x\circ\iscripthat _y)}=\rmtr\sqbrac{\rho\sum _x(a_x\circ\iscripthat _y)}=\rmtr\sqbrac{\rho (\,\iscripthat\mid A)_y}
\end{align*}
We conclude that $(\iscript\mid A)^\wedge =(\,\iscripthat\mid A)$.
\end{proof}

We have seen that if $A,B\in\rmob (H)$, then $A\circ\lscript ^B=\lscript ^A\circ\lscript ^B$. On the other hand
\begin{equation*}
(\lscript ^A\circ B)_{(x,y)}=\lscript _x^A\circ b_y=a_x^{1/2}b_ya_x^{1/2}=a_x\circ b_y=(A\circ B)_{(x,y)}
\end{equation*}
Hence, $\lscript ^A\circ B=A\circ B$. We now treat trivial and semi-trivial instruments. 

\begin{thm}    
\label{thm42}
{\rm{(i)}}\enspace Let $a\in\escript (H)$ and let $\iscript$ be a trivial operation $\iscript (\rho )=\rmtr (\rho b)\alpha$. Then
$\iscript\circ A=\rmtr (\alpha a)b$ and $a\circ\iscript$ is trivial with state $\alpha$ and effect $a\circ b$.
{\rm{(ii)}}\enspace Let $A=\brac{a_x\colon x\in\Omega _A}$ be an observable and let $\iscript$ be semi-trivial with states $\alpha _y$ and observable
$B$. Then $(\iscript\circ A)_{(x,y)}=\rmtr (\alpha _xa_y)b_x$ and $A\circ\iscript$ is semi-trivial with states $\alpha _y$ and observable $A\circ B$.
{\rm{(iii)}}\enspace Let $A=\brac{a_x\colon x\in\Omega _A}$ be an observable and let $\iscript$ be trivial with state $\alpha$ and observable $B$. Then
$(\iscript\circ a)_{(x,y)}=\rmtr (\alpha a_y)b_x$ and $A\circ\iscript$ is trivial with state $\alpha$ and observable $A\circ B$.
\end{thm}
\begin{proof}
(i)\enspace Letting $\iscript$ have Kraus decomposition $\iscript (\rho )=\sum A_i\rho A_i^*$ with $\sum A_i^*A_i\le I$ we have for all $\rho\in\sscript (H)$ that
\begin{align*}
\rmtr (\rho\iscript\circ a)&=\rmtr\paren{\rho\sum A_i^*aA_i}=\sum\rmtr (\rho A_i^*aA_i)=\sum\rmtr (A_i\rho A_i^*a)\\
   &=\rmtr\paren{\sum A_i\rho A_i^*a}=\rmtr\sqbrac{\iscript (\rho )a}=\rmtr\sqbrac{\rmtr (\rho b)\alpha a}\\
   &=\rmtr (\rho b)\rmtr (\alpha a)=\rmtr\sqbrac{\rho\rmtr (\alpha a)b}
\end{align*}
Hence, $\iscript\circ a=\rmtr (\alpha a)b$. Moreover, for all $\rho\in\sscript (H)$ we obtain
\begin{align*}
(a\circ\iscript )(\rho )&=\iscript (a^{1/2}\rho a^{1/2})=\rmtr (a^{1/2}\rho a^{1/2}b)\alpha =\rmtr (\rho a^{1/2}ba^{1/2})\alpha\\
   &=\rmtr (\rho a\circ b)\alpha
\end{align*}
Thus, $a\circ\iscript$ is trivial with state $\alpha$ and effect $a\circ b$.\newline
(ii)\enspace Letting $\iscript _x$ have Kraus decomposition $\iscript _x(\rho )=\sum A_i^x\rho (A_i^x)^*$, we have from (i) that
\begin{equation*}
\rmtr\sqbrac{\rho (\iscript\circ A)_{(x,y)}}=\rmtr (\rho\iscript _x\circ a_y)=\rmtr\sqbrac{\rmtr (\alpha _xa_y)b_x}
\end{equation*}
Therefore, $(\iscript\circ A)_{(x,y)}=\rmtr (\alpha _xa_y)b_x$. Moreover, by (i) we have for all $\rho\in\sscript (H)$ that
\begin{equation*}
(A\circ\iscript )_{(x,y)}(\rho )=(a_x\circ\iscript _y)(\rho )=\rmtr (\rho a_x\circ b_y)\alpha _y
\end{equation*}
Hence, $A\circ\iscript$ is semi-trivial with states $\alpha _y$ and observable $(A\circ B)_{(x,y)}=a_x\circ b_y$.
(iii) follows from (ii).
\end{proof}

Let $\iscript\in\oscript (H)$ be a Kraus with $\iscript (\rho )=S\rho S^*$ and let $a\in\escript (H)$. Then $\iscript\circ a=S^*aS$ and
\begin{equation*}
(a\circ\iscript )(\rho )=\iscript (a^{1/2}\rho a^{1/2})=Sa^{1/2}\rho a^{1/2}S^*
\end{equation*}
Thus, $a\circ\iscript\in\oscript (H)$ is Kraus with Kraus operator $Sa^{1/2}$. More generally, if $\iscript\in\rmin (H)$ is Kraus with
$\iscript _x(\rho )=S_x\rho S_x^*$ and $A=\brac{a_y\colon y\in\Omega _A}\in\rmob (H)$, then
\begin{align*}
(\iscript\circ A)_{(x,y)}&=\iscript _x\circ a_y=S_x^*a_yS_x\\
\intertext{and}
(A\circ\iscript )_{(x,y)}(\rho )&=(a_y\circ\iscript _x)(\rho )=\iscript _x(a_y^{1/2}\rho a _y^{1/2})=S_xa_y^{1/2}\rho a_y^{1/2}S_x^*\\
\end{align*}
Thus, $A\circ\iscript\in\rmin (H)$ is Kraus with Kraus operators $S_xa_y^{1/2}$.

\begin{exam}{10}  
If $\iscript\in\oscript (H)$ is a channel, it is easy to check that $J(a)=\iscript\circ a$ is a convex, effect algebra morphism on $\escript (H)$. We now show that $J$ need not be a monomorphism. That is, if $J(a)\perp J(b)$, then we need not have $a\perp b$. Let $\brac{\psi _1,\psi _2}$ be an orthonormal basis for $\complex ^2$ and let $a_1=\ket{\psi _1}\bra{\psi _1}$, $a_2=\ket{\psi _2}\bra{\psi _2}$. Then $\iscript (\rho )=a_1\rho a_1+a_2\rho a_2$ is a channel. Letting $d=\ket{\psi _1+\psi _2}\bra{\psi _1+\psi _2}$ we have that
\begin{equation*}
d=2\ket{\frac{\psi _1+\psi _2}{\sqrt{2}}}\bra{\frac{\psi _1+\psi _2}{\sqrt{2}}}
\end{equation*}
Hence, $d$ is twice a one-dimensional projection so $d\not\le I$. Letting $a=b=\tfrac{1}{2}d$ we have that $a,b\in\escript (H)$ and $a+b=d\not\le I$ so
$a\not\perp b$. However,
\begin{align*}
J(a)+J(b)&=J(d)\\
  &=\ket{\psi _1}\bra{\psi _1}\,\ket{\psi _1+\psi _2}\bra{\psi _1+\psi _2}\,\ket{\psi _1}\bra{\psi _1}\\
  &\qquad +\ket{\psi _2}\bra{\psi _2}\,\ket{\psi _1+\psi _2}\bra{\psi _1+\psi _2}\,\ket{\psi _2}\bra{\psi _2}\\
   &=\ket{\psi _1}\bra{\psi _1}+\ket{\psi _2}\bra{\psi _2}=I
\end{align*}
so $J(a)\perp J(b)$.\hfill\qedsymbol
\end{exam}

An observable $B=\brac{b_y\colon y\in\Omega _B}$ is \textit{part} \cite{fhl18} of an observable $A=\brac{a_x\colon x\in\Omega _A}$ if there exists a surjection $f\colon\Omega _A\to\Omega _B$ such that
\begin{equation*}
b_y=A_{f^{-1}(y)}=\sum\brac{a_x\colon f(x)=y}
\end{equation*}
We then write $B=f(A)$. Two observables $B,C$ \textit{coexist} \cite{bgl95,hz12,hrsz09} if there exists an observable $A$ such that $B=f(A)$, $C=g(A)$. Thus, $B$ and $C$ coexist if they can be measured by applying a single observable $A$. If $A,B\in\rmob (H)$ with
$A=\brac{a_x\colon x\in\Omega _A}$, $B=\brac{b_y\colon y\in\Omega _B}$, define $f\colon\Omega _A\times\Omega _B\to\Omega _B$ by $f(x,y)=y$. Then
\begin{equation*}
(B\mid A)_y=\sum _xa_x\circ b_y=\sum _x(A\circ B)_{(x,y)}=\sum\brac{(A\circ B)_{(x,y)}\colon f(x,y)=y}
\end{equation*}
Therefore, $(B\mid A)_y=(A\circ B)_{f^{-1}(y)}$ so $(B\mid A)=f(A\circ B)$. We conclude that $(B\mid A)$ and $A\circ B$ coexist. Also, $A$ and
$A\circ B$ coexist because $a_x=\sum _ya_x\circ b_y$. Thus, $(B\mid A)$ and $A$ coexist.

If $\iscript ,\jscript\in\rmin (H)$, we have the instrument $\iscript\circ\jscript$ given by
 $(\iscript\circ\jscript )_{(x,y)}(\rho )=\jscript _y\paren{\iscript _x(\rho )}$ with channel ${\overline{\iscript\circ\jscript}}=\iscriptbar\circ\jscriptbar$ and instrument $(\jscript\mid\iscript )$ given by $(\jscript\mid\iscript )_y(\rho )=\jscript _y\paren{\,\iscriptbar (\rho )}$ with channel
${\overline{(\jscript\mid\iscript )}}=\iscriptbar\circ\jscriptbar$. We say that $\jscript$ is \textit{part} of $\iscript$ if there exists a surjection
$f\colon\Omega _\iscript\to\Omega _\jscript$ such that
\begin{equation*}
\jscript _y=\iscript _{f^{-1}(y)}=\sum\brac{\iscript _x\colon f(x)=y}
\end{equation*}
We then write $\jscript =f(\iscript )$. As with observables, we say that $\jscript ,\kscript\in\rmin (H)$ \textit{coexist} if there exists an instrument $\iscript$ such that $\jscript =f(\iscript )$, $\kscript =g(\iscript )$. Moreover, we have that $(\jscript\mid\iscript )$ and $\iscript\circ\jscript$ coexist. However,
$\iscript _x\ne\sum _y\iscript _x\circ\jscript _y$ in general, so $(\jscript\mid\iscript )$ and $\iscript$ may not coexist.

\begin{lem}    
\label{lem43}
{\rm{(i)}}\enspace $f(\iscript )^\wedge =f(\,\iscripthat\,)$.
{\rm{(ii)}}\enspace If $\jscript$, $\kscript$ coexist, the $\jscripthat$, $\kscripthat$ coexist.
\end{lem}
\begin{proof}
(i)\enspace For all $\rho\in\sscript (H)$ we have that
\begin{align*}
\rmtr\sqbrac{\rho f(\iscript _y^\wedge}&=\rmtr\sqbrac{f(\iscript )_y(\rho )}=\rmtr\sqbrac{\iscript _{f^{-1}(y)}(\rho )}
   =\rmtr\sqbrac{\sum _{f(x)=y}\iscript _x(\rho )}\\
   &=\sum _{f(x)=y}\rmtr\sqbrac{\iscript _x(\rho )}=\sum _{f(x)=y}\rmtr(\rho\iscripthat _x)=\rmtr\paren{\rho\sum _{f(x)=y}\iscripthat _x}
   =\rmtr\sqbrac{\rho f(\,\iscripthat\,)_y}
\end{align*}
Hence, $f(\iscript )^\wedge =f(\,\iscripthat\,)$.
(ii)\enspace If $\jscript$, $\kscript$ coexist, the $\jscript =f(\iscript )$, $\kscript =g(\iscript )$ for some $\iscript\in\rmin (H)$. By (i) we obtain
$\jscripthat =f(\,\iscripthat\,)$ and $\kscripthat =g(\,\iscripthat\,)$ so $\jscripthat$ and $\kscripthat$ coexist.
\end{proof}

\begin{lem}    
\label{lem44}
{\rm{(i)}}\enspace If $\iscript _x(\rho )=\rmtr (\rho a_x)\alpha$ is trivial and $\jscript =f(\iscript )$, then $\jscript$ is trivial with state $\alpha$ and observable $f(A)$.
{\rm{(ii)}}\enspace If $\jscript$, $\kscript$ are trivial with the same state $\alpha$ and $\jscripthat$, $\kscripthat$ coexist, then $\jscript$, $\kscript$ coexist.
\end{lem}
\begin{proof}
(i)\enspace For all $\rho\in\sscript (H)$ we obtain
\begin{equation*}
\jscript _y(\rho )=\sum _{f(x)=y}\iscript _x(\rho )=\sum _{f(x)=y}\rmtr (\rho a_x)\alpha =\rmtr\paren{\rho\sum _{f(x)=y}a_x}\alpha
   =\rmtr\sqbrac{\rho f(A)_y}\alpha
\end{equation*}
Hence, $\jscript$ is trivial with state $\alpha$ and observable $f(A)$.
(ii)\enspace Let $\jscript _y(\rho )=\rmtr (\rho b_y)\alpha$, $\kscript _z(\rho )=\rmtr (\rho c_z)\alpha$ be trivial with the same state $\alpha$. Since
$\jscript$, $\kscript$ coexist there exists an observable $A=\brac{a_x\colon x\in\Omega _A}$ such that $\jscripthat =f(A)$, $\kscripthat =g(A)$. Letting
$\iscript\in\rmin (H)$ be defined by $\iscript _x(\rho )=\rmtr (\rho a_x)\alpha$ we obtain from (i) that
\begin{align*}
\jscript _y&=\rmtr (\rho\jscripthat _y)\alpha =\rmtr\sqbrac{\rho f(A)_y}\alpha =f(\iscript )_y(\rho )\\
\intertext{and}
\kscript _z(\rho )&=\rmtr (\rho\kscripthat _z)\alpha =\rmtr\sqbrac{\rho g(A)_z}\alpha =g(\iscript )_z(\rho )\\
\end{align*}
Hence, $\jscript =f(\iscript )$, $\kscript =g(\iscript )$ so $\jscript$, $\kscript$ coexist
\end{proof}

\end{document}